\documentclass{article}

\usepackage{hyperref}
\usepackage{amsmath}
\usepackage{amssymb}
\usepackage{algorithmic}
\usepackage{dsfont}  
\usepackage{amsthm}
\usepackage{fullpage}

\newtheorem{theorem}{Theorem}
\newtheorem{lemma}[theorem]{Lemma}

\newcommand{\abs}[1]{\left|#1\right|}
\newcommand{\eps}{\epsilon}
\newcommand{\1}{\mathds{1}}
\newcommand{\N}{\mathbb{N}}
\newcommand{\inner}[2]{\langle #1, #2\rangle}
\newcommand{\ones}{\mathbf{1}}

\newcommand{\R}{\mathbb{R}}
\newcommand{\E}{\mathbb{E}}
\newcommand{\etal}{et al.\ }
\newcommand{\floor}[1]{\lfloor #1 \rfloor}

\mathchardef\mhyp="2D
\newcommand{\apxpips}{\mathsf{round}\mhyp\mathsf{and}\mhyp\mathsf{alter}\mhyp\mathsf{by}\mhyp\mathsf{sorting}}
\newcommand{\apxvio}{\mathsf{round}\mhyp\mathsf{alter}\mhyp\mathsf{small}\mhyp\mathsf{width}}

\begin{document}

\title{$\ell_1$-sparsity Approximation Bounds for Packing Integer Programs
  \thanks{C.~Chekuri and K.~Quanrud supported in part by NSF grant
    CCF-1526799. M.~Torres supported in part by fellowships from NSF
    and the Sloan Foundation. University of Illinois, Urbana-Champaign, IL
    61801. \texttt{\{chekuri, quanrud2, manuelt2\}@illinois.edu}}}

\author{Chandra Chekuri \and Kent Quanrud \and Manuel R.\ Torres}

\maketitle

\begin{abstract}
  We consider approximation algorithms for packing integer programs
  (PIPs) of the form
  $\max\{\inner{c}{x} : Ax \le b, x \in \{0,1\}^n\}$ where $c$, $A$,
  and $b$ are nonnegative. We let $W = \min_{i,j} b_i / A_{i,j}$
  denote the width of $A$ which is at least $1$. Previous work by
  Bansal et al. \cite{bansal-sparse} obtained an
  $\Omega(\frac{1}{\Delta_0^{1/\lfloor W \rfloor}})$-approximation
  ratio where $\Delta_0$ is the maximum number of nonzeroes in any
  column of $A$ (in other words the $\ell_0$-column sparsity of
  $A$). They raised the question of obtaining approximation ratios
  based on the $\ell_1$-column sparsity of $A$ (denoted by $\Delta_1$)
  which can be much smaller than $\Delta_0$.  Motivated by recent work
  on covering integer programs (CIPs) \cite{cq,chs-16} we show that
  simple algorithms based on randomized rounding followed by
  alteration, similar to those of Bansal et al.\ \cite{bansal-sparse}
  (but with a twist), yield approximation ratios for PIPs based on
  $\Delta_1$.  First, following an integrality gap example from
  \cite{bansal-sparse}, we observe that the case of $W=1$ is as hard
  as maximum independent set even when $\Delta_1 \le 2$.  In sharp
  contrast to this negative result, as soon as width is strictly
  larger than one, we obtain positive results via the natural LP relaxation. 
  For PIPs with width
  $W = 1 + \eps$ where $\eps \in (0,1]$, we obtain an
  $\Omega(\eps^2/\Delta_1)$-approximation. In the large width regime,
  when $W \ge 2$, we obtain an
  $\Omega((\frac{1}{1 + \Delta_1/W})^{1/(W-1)})$-approximation. We
  also obtain a $(1-\eps)$-approximation when
  $W = \Omega(\frac{\log (\Delta_1/\eps)}{\eps^2})$.
\end{abstract}

\section{Introduction}

Packing integer programs (abbr.\ PIPs) are an expressive class of
integer programs of the form:
\[
  \text{maximize} \ \inner{c}{x} \ \text{over} \ x \in \{0,1\}^n \ \text{s.t.} \  Ax \le b,
\]
where $A \in \R_{\ge 0}^{m \times n}$, $b \in \R_{\ge 0}^m$ and
$c \in \R_{\ge 0}^n$ all have nonnegative entries\footnote{We can allow
  the variables to have general integer upper bounds instead of
  restricting them to be boolean. As observed in \cite{bansal-sparse}, one can
  reduce this more general case to the $\{0,1\}$ case without too much
  loss in the approximation.}.  Many important
problems in discrete and combinatorial optimization can be cast as
special cases of PIPs. These include the maximum independent set in
graphs and hypergraphs, set packing, matchings and $b$-matchings,
knapsack (when $m=1$), and the multi-dimensional knapsack.
The maximum independent set problem (MIS), a special case of PIPs, is
NP-hard and unless $P=NP$ there is no $n^{1-\eps}$-approximation where
$n$ is the number of nodes in the graph
\cite{hastad-MIS,zuckerman-MIS}. For this reason it is meaningful to
consider special cases and other parameters that control the
difficulty of PIPs. Motivated by the fact that MIS admits a simple
$\frac{1}{\Delta(G)}$-approximation where $\Delta(G)$ is the maximum
degree of $G$, previous work considered approximating PIPs based on the
maximum number of nonzeroes in any column of $A$ (denoted by $\Delta_0$); note that
when MIS is written as a PIP, $\Delta_0$ coincides with $\Delta(G)$.
As another example, when maximum weight matching is written as a PIP,
$\Delta_0 = 2$. Bansal \etal \cite{bansal-sparse} obtained a simple
and clever algorithm that achieved an
$\Omega(1/\Delta_0)$-approximation for PIPs via the natural LP
relaxation; this improved previous work of Pritchard
\cite{Pritchard09,PritchardC11} who was the first to obtain an
approximation for PIPs only as a function of $\Delta_0$. Moreover, the
rounding algorithm in \cite{bansal-sparse} can be viewed as a
contention resolution scheme which allows one to get similar
approximation ratios even when the objective is submodular
\cite{bansal-sparse,cvz-crs}. It is well-understood that PIPs become
easier when the entries in $A$ are small compared to the packing
constraints $b$. To make this quantitative we consider the
well-studied notion called the \emph{width}
defined as $W := \min_{i,j: A_{i,j} > 0} b_i/A_{i,j}$. Bansal \etal
obtain an $\Omega( (\frac{1}{\Delta_0})^{1/\floor{W}})$-approximation
which improves as $W$ becomes larger. Although they do not state it
explicitly, their approach also yields a $(1-\eps)$-approximation when
$W = \Omega(\frac{1}{\eps^2}\log (\Delta_0/\eps))$.

$\Delta_0$ is a natural measure for combinatorial applications such as
MIS and matchings where the underlying matrix $A$ has entries from
$\{0,1\}$. However, in some applications of PIPs such as knapsack and
its multi-dimensional generalization which are more common in
resource-allocation problems, the entries of $A$ are arbitrary
rational numbers (which can be assumed to be from the interval $[0,1]$
after scaling).  In such applications it is natural to consider
another measure of column-sparsity which is based on the $\ell_1$
norm. Specifically we consider $\Delta_1$, the maximum column sum of
$A$. Unlike $\Delta_0$, $\Delta_1$ is not scale invariant so one needs
to be careful in understanding the parameter and its relationship to
the width $W$. For this purpose we normalize the constraints
$Ax \le b$ as follows. Let $W = \min_{i,j: A_{i,j} > 0} b_i/A_{i,j}$
denote the width as before (we can assume without loss of generality
that $W \ge 1$ since we are interested in integer solutions). We can
then scale each row $A_i$ of $A$ separately such that, after scaling, the $i$'th
constraint reads as $A_i x \le W$.  After scaling all rows in this
fashion, entries of $A$ are in the interval $[0,1]$, and the
maximum entry of $A$ is equal to $1$. Note that this scaling process
does not alter the original width. We let $\Delta_1$ denote the
maximum column sum of $A$ after this normalization and observe that
$1 \le \Delta_1 \le \Delta_0$. In many settings of interest
$\Delta_1 \ll \Delta_0$. We also observe that $\Delta_1$ is a
more robust measure than $\Delta_0$; small perturbations of the
entries of $A$ can dramatically change $\Delta_0$ while $\Delta_1$
changes minimally.

Bansal \etal raised the question of obtaining an approximation ratio for
PIPs as a function of only $\Delta_1$. They observed that this is not
feasible via the natural LP relaxation by describing a simple example
where the integrality gap of the LP is $\Omega(n)$ while $\Delta_1$ is
a constant. In fact their example essentially shows the existence of a
simple approximation preserving reduction from MIS to PIPs such that
the resulting instances have $\Delta_1 \le 2$; thus no approximation
ratio that depends only on $\Delta_1$ is feasible for PIPs unless $P=NP$. These
negative results seem to suggest that pursuing bounds based on
$\Delta_1$ is futile, at least in the worst case. However, the
starting point of this paper is the observation that both the
integrality gap example and the hardness result are based on instances
where the width $W$ of the instance is arbitrarily close to $1$. We
demonstrate that these examples are rather brittle and obtain several
positive results when we consider $W \ge (1+\eps)$ for any fixed
$\eps > 0$.

\subsection{Our results}
Our first result is on the hardness of approximation for PIPs that we
already referred to. The hardness result suggests that one
should consider instances with $W > 1$. Recall that after normalization we have
$\Delta_1 \ge 1$ and $W \ge 1$ and the maximum entry of $A$ is $1$.
We consider three regimes of $W$ and obtain the following results, all
via the natural LP relaxation, which also establish corresponding
upper bounds on the integrality gap.
\begin{itemize}
\item[(i)] $1 <W \le 2$. For $W = 1+\eps$ where $\eps \in (0,1]$ we
  obtain an $\Omega(\frac{\eps^2}{\Delta_1})$-approximation.
  \item[(ii)] $W \ge 2$. We obtain an $\Omega( (\frac{1}{1+
      \frac{\Delta_1}{W}})^{1/(W-1)})$-approximation which can be
    simplified to $\Omega((\frac{1}{1+ \Delta_1})^{1/(W-1)})$ since $W
    \ge 1$.
  \item[(iii)] A $(1-\eps)$-approximation when $W = \Omega(\frac{1}{\eps^2}\log (\Delta_1/\eps))$.
\end{itemize}
Our results establish approximation bounds based on $\Delta_1$ that
are essentially the same as those based on $\Delta_0$ as long as the
width is not too close to $1$.
We describe randomized algorithms which can be derandomized
via standard techniques. The algorithms can be viewed as
contention resolution schemes, and via known techniques
\cite{bansal-sparse,cvz-crs}, the results yield corresponding
approximations for submodular objectives; we omit these extensions in
this version.

All our algorithms are based on a simple randomized rounding plus
alteration framework that has been successful for both packing and
covering problems. Our scheme is similar to that of Bansal \etal
at a high level but we make a simple but important change in the
algorithm and its analysis. This is inspired by recent work on covering
integer programs \cite{cq} where $\ell_1$-sparsity based approximation
bounds from \cite{chs-16} were simplified.

\subsection{Other related work}
We note that PIPs are equivalent to the multi-dmensional knapsack problem.
When $m=1$ we have the classical knapsack problem which admits a very
efficient FPTAS (see \cite{chan-knapsack}). There is a PTAS for any fixed $m$
\cite{FriezeC84} but unless $P=NP$ an FPTAS does not exist
for $m=2$.

Approximation algorithms for PIPs in their general form were
considered initially by Raghavan and Thompson \cite{RaghavanT87} and
refined substantially by Srinivasan \cite{Srinivasan99}. Srinivasan
obtained approximation ratios of the form $\Omega(1/n^{W})$ when $A$
had entries from $\{0,1\}$, and a ratio of the form
$\Omega(1/n^{1/\floor{W}})$ when $A$ had entries from
$[0,1]$. Pritchard \cite{Pritchard09} was the first to obtain a bound
for PIPs based solely on the column sparsity parameter $\Delta_0$.  He
used iterated rounding and his initial bound was improved in
\cite{PritchardC11} to $\Omega(1/\Delta_0^2)$.  The current state of
the art is due to Bansal \etal \cite{bansal-sparse}. Previously we
ignored constant factors when describing the ratio. In fact
\cite{bansal-sparse} obtains a ratio of
$(1 - o(1)\frac{e-1}{e^2\Delta_0})$ by strengthening the basic LP
relaxation.

In terms of hardness of approximation, PIPs generalize MIS and hence
one cannot obtain a ratio better than $n^{1-\eps}$ unless $P=NP$
\cite{hastad-MIS,zuckerman-MIS}. Building on MIS, \cite{CK04} shows that PIPs
are hard to approximate within a $n^{\Omega(1/W)}$ factor for any constant width
$W$. Hardness of MIS in bounded degree graphs \cite{t-01} and hardness for
$k$-set-packing \cite{hss-06} imply that PIPs are hard to
approximate to within $\Omega(1/\Delta_0^{1-\eps})$ and to within
$\Omega((\log \Delta_0)/\Delta_0)$ when $\Delta_0$ is a
sufficiently large constant. These hardness results are based on
$\{0,1\}$ matrices for which $\Delta_0$ and $\Delta_1$ coincide.

There is a large literature on deterministic and randomized rounding
algorithms for packing and covering integer programs and connections
to several topics and applications including discrepancy theory.
$\ell_1$-sparsity guarantees for covering integer programs were first
obtained by Chen, Harris and Srinivasan \cite{chs-16} partly
inspired by \cite{h-15}.

\section{Hardness of approximating PIPs as a function of $\Delta_1$}
\label{sec:hardness}
Bansal \etal \cite{bansal-sparse} showed that the integrality gap of
the natural LP relaxation for PIPs is $\Omega(n)$ even when $\Delta_1$
is a constant. One can use essentially the same construction to show the following
theorem.

\begin{theorem}\label{thm:hardness}
  There is an approximation preserving reduction from MIS to
  instances of PIPs with $\Delta_1 \le 2$.
\end{theorem}

    \begin{proof}
      Let $G = (V,E)$ be an undirected graph without self-loops and let
      $n = \abs{V}$. Let $A \in [0,1]^{n \times n}$ be indexed by $V$. For
      all $v \in V$, let $A_{v,v} = 1$. For all $uv \in E$, let
      $A_{u,v} = A_{v,u} = 1/n$. For all the remaining entries in $A$ that
      have not yet been defined, set these entries to $0$.
      Consider the following PIP:
      \begin{equation}\label{pr:ind}
        \text{maximize}\ \inner{x}{\ones} \ 
          \text{over} \ x \in \{0,1\}^n \ \text{s.t.} \ Ax \le 1.
      \end{equation}
      
      Let $S$ be the set of all feasible integral solutions
      of~(\ref{pr:ind}) and $\mathcal{I}$ be the set of independent sets of
      $G$. Define $g : S \to \mathcal{I}$ where $g(x) = \{v : x_v = 1\}$. To
      show $g$ is surjective, consider a set $I \in \mathcal{I}$. Let $y$ be
      the characteristic vector of $I$. That is, $y_v$ is $1$ if $v \in I$
      and $0$ otherwise. Consider the row in $A$ corresponding to an
      arbitrary vertex $u$ where $y_u = 1$. For all $v \in V$ such that $v$
      is a neighbor to $u$, $y_v = 0$ as $I$ is an independent set. 
      Thus, as the nonzero entries in $A$ of the row corresponding to 
      $u$ are, by construction, the neighbors of $u$, it follows that 
      the constraint corresponding to $u$ is satisfied in~(\ref{pr:ind}). 
      As $u$ is an arbitrary vertex, it follows that $y$ is a feasible 
      integral solution to~(\ref{pr:ind}) and as $I = \{v : y_v = 1\}$, 
      $g(y) = I$.
      
      Define $h : S \to \N_0$ such that $h(x) = \abs{g(x)}$. It is clear
      that $\max_{x \in S} h(x)$ is equal to the optimal value
      of~(\ref{pr:ind}). Let $I_{max}$ be a maximum independent set of
      $G$. As $g$ is surjective, there exists $z \in S$ such that
      $g(z) = I_{max}$. Thus, $\max_{x\in S} h(x) \ge \abs{I_{max}}$. As
      $\max_{x \in S} h(x)$ is equal to the optimum value of~(\ref{pr:ind}),
      it follows that a $\beta$-approximation for PIPs implies a
      $\beta$-approximation for maximum independent set.
      
      Furthermore, we note that for this PIP, $\Delta_1 \le 2$, 
      thus concluding the proof.
    \end{proof}

Unless $P=NP$, MIS does not admit a $n^{1-\eps}$-approximation for any
fixed $\eps > 0$ \cite{hastad-MIS,zuckerman-MIS}. Hence the preceding
theorem implies that unless $P=NP$ one cannot obtain an approximation
ratio for PIPs solely as a function of $\Delta_1$.

\section{Round and alter framework}
The algorithms in this paper have the same high-level structure. The
algorithms first scale down the fractional solution $x$ by some factor
$\alpha$, and then randomly round each coordinate independently.
The rounded solution $x'$ may not be feasible for the constraints.
The algorithm alters $x'$ to a feasible $x''$ by considering
each constraint separately in an arbitrary order; if $x'$ is not
feasible for constraint $i$ some subset $S$ of variables are chosen
to be set to $0$.  Each constraint corresponds
to a knapsack problem and the framework (which is adapted from
\cite{bansal-sparse}) views the problem as the intersection of
several knapsack constraints. A
formal template is given in Figure~\ref{fig:framework}. To make the
framework into a formal algorithm, one must define $\alpha$ and how to
choose $S$ in the for loop. These parts will depend on the regime of
interest.

\begin{figure}[t]
\centering
\fbox{\parbox{0.76\linewidth}
{
Round-and-Alter Framework: input $A$, $b$, and $\alpha$
\begin{algorithmic}
\STATE let $x$ be the optimum fractional solution of the natural LP relaxation
\STATE for $j \in [n]$, set $x_j'$ to be $1$ independently with probability $\alpha x_j$ and $0$ otherwise
\STATE $x'' \leftarrow x'$
\FOR{$i \in [m]$}
	\STATE find $S\subseteq [n]$ such that setting $x_j' = 0$ for all $j \in S$ would satisfy $\inner{e_i}{Ax'} \le b_i$
	\STATE for all $j \in S$, set $x_j'' = 0$
\ENDFOR
\RETURN{$x''$}
\end{algorithmic}
}}
\caption{Randomized rounding with alteration framework.}
\label{fig:framework}
\end{figure}

For an algorithm that follows the round-and-alter framework, the expected output of the algorithm is
$\E\left[\inner{c}{x''}\right] = \sum_{j =1}^n c_j\cdot\Pr[x_j'' = 1]$.
Independent of how $\alpha$ is defined or how $S$ is chosen, $\Pr[x_j'' = 1] = \Pr[x_j'' = 1 | x_j' = 1]\cdot \Pr[x_j' = 1]$ since $x_j'' \le x_j'$. Then we have
\[
\E[\inner{c}{x''}] = \alpha \sum_{j=1}^n c_jx_j\cdot \Pr[x_j'' = 1 | x_j' = 1].
\]
Let $E_{ij}$ be the event that $x_j''$ is set to $0$ when ensuring constraint $i$ is satisfied in the for loop. As $x_j''$ is only set to $0$ if at least one constraint sets $x_j''$ to $0$, we have
\[
\Pr[x_j'' = 0 | x_j ' = 1] = \Pr\left[\bigcup_{i \in [m]} E_{ij} | x_j' = 1\right] \le \sum_{i =1}^m \Pr[E_{ij} | x_j' = 1].
\]

Combining these two observations, we have the following lemma, which applies to all of our subsequent algorithms.
\begin{lemma}\label{lem:general}
Let $\mathcal{A}$ be a randomized rounding algorithm that follows the round-and-alter framework given in Figure~\ref{fig:framework}. Let $x'$ be the rounded solution obtained with scaling factor $\alpha$. Let $E_{ij}$ be the event that $x_j''$ is set to $0$ by constraint $i$. If for all $j \in [n]$ we have
$\sum_{i =1}^m \Pr[E_{ij} | x_j' = 1] \le \gamma,$
then $\mathcal{A}$ is an $\alpha(1 - \gamma)$-approximation for PIPs.
\end{lemma}

We will refer to the quantity $\Pr[E_{ij}|x_j' = 1]$ as the \emph{rejection probability} of item $j$ in constraint $i$. We will also say that constraint $i$ \emph{rejects} item $j$ if $x_j''$ is set to $0$ in constraint $i$.

\section{The large width regime: $W \ge 2$}\label{sec:W2}

In this section, we consider PIPs with width $W \ge 2$.  Recall that
we assume $A \in [0,1]^{m \times n}$ and $b_i = W$ for all $i \in
[m]$. Therefore we have $A_{i,j} \le W/2$ for all $i,j$ and from a
knapsack point of view all items are ``small''.  We apply the
round-and-alter framework in a simple fashion where in each constraint
$i$ the coordinates are sorted by the coefficents in that row and the
algorithm chooses the largest prefix of coordinates that fit in the
capacity $W$ and the rest are discarded. We emphasize that this
sorting step is crucial for the analysis and differs from the scheme
in \cite{bansal-sparse}. Figure~\ref{fig:pseudo-W2} describes the formal
algorithm.

\begin{figure}[t]
\centering
\fbox{\parbox{0.76\linewidth}
{
$\apxpips(A, b, \alpha_1)$:
\begin{algorithmic}
\STATE let $x$ be the optimum fractional solution of the natural LP relaxation
\STATE for $j \in [n]$, set $x_j'$ to be $1$ independently with probability $\alpha_1 x_j$ and $0$ otherwise
\STATE $x'' \leftarrow x'$
\FOR{$i \in [m]$}
	\STATE sort and renumber such that $A_{i,1} \le \cdots\le A_{i,n}$
		\STATE $s \leftarrow \max\{\ell \in [n] : \sum_{j = 1}^\ell A_{i,j}x_j' \le b_i\}$
		\STATE for each $j \in [n]$ such that $j > s$, set $x_j'' = 0$
\ENDFOR
\RETURN{$x''$}
\end{algorithmic}
}}
\caption{Round-and-alter in the large width regime. Each constraint
sorts the coordinates in increasing size and greedily picks a feasible set
and discards the rest.
}
\label{fig:pseudo-W2}
\end{figure}

\paragraph{The key property for the analysis:} The analysis relies on
obtaining a bound on the rejection probability of coordinate
$j$ by constraint $i$. Let $X_j$ be the indicator variable for $j$
being chosen in the first step.  We show that $\Pr[E_{ij} \mid X_j =
1] \le c A_{ij}$ for some $c$ that depends on the scaling factor
$\alpha$.  Thus coordinates with smaller coefficients are less likely
to be rejected.  The total rejection probability of $j$, $\sum_{i=1}^m
\Pr[E_{ij}\mid X_j=1]$, is proportional to the column sum of coordinate
$j$ which is at most $\Delta_1$.

The analysis relies on the Chernoff bound, and
depending on the parameters, one needs to adjust the analysis.
In order to highlight the main ideas we provide a detailed
proof for the simplest case and include the proofs of the
other cases in the appendix.

\subsection{An $\Omega(1/\Delta_1)$-approximation algorithm}\label{sec:W2-weak}

We show that $\apxpips$ yields an $\Omega(1/\Delta_1)$-approximation
if we set the scaling factor $\alpha_1 = \frac{1}{c_1\Delta_1}$ where
$c_1 = 4e^{1+1/e}$.

The rejection probability is captured by the following main lemma.
\begin{lemma}\label{lem:W2-weak}
Let $\alpha _1= \frac{1}{c_1\Delta_1}$ for $c_1 = 4e^{1+1/e}$. Let $i \in [m]$ and $j \in [n]$. Then we have $\Pr[E_{ij} | X_j = 1] \le \frac{A_{i,j}}{2\Delta_1}$ in the algorithm $\apxpips(A, b, \alpha_1)$.
\end{lemma}
\begin{proof}
  At iteration $i$ of $\apxpips$, after the set $\{A_{i,1},
  \ldots, A_{i,n}\}$ is sorted, the indices are renumbered so that
  $A_{i,1} \le \cdots \le A_{i,n}$. Note that $j$ may now be a
  different index $j'$, but for simplicity of notation we will refer
  to $j'$ as $j$. Let $\xi_\ell = 1$ if $x_\ell' = 1$ and $0$
  otherwise. Let $Y_{ij} = \sum_{\ell = 1}^{j-1} A_{i,\ell}\xi_\ell$.

  If $E_{ij}$ occurs, then $Y_{ij} > W - A_{i,j}$, since $x_j''$ would
  not have been set to zero by constraint $i$ otherwise. That is,
\[
\Pr[E_{ij} | X_j = 1] \le \Pr[Y_{ij} > W - A_{i,j} | X_j = 1].
\]
The event $Y_{ij} >W - A_{i,j}$ does not depend on $x_j'$. Therefore,
\[
\Pr[Y_{ij} > W - A_{i,j} | X_j = 1] \le  \Pr[Y_{ij} \ge W - A_{i,j}].
\]

To upper bound $\E[Y_{ij}]$,  we have
\[
\E[Y_{ij}] = \sum_{\ell = 1}^{j-1} A_{i,\ell} \cdot \Pr[X_\ell = 1] \le \alpha_1\sum_{\ell = 1}^n  A_{i,\ell}x_\ell \le \alpha_1W.
\]
As $A_{i,j} \le 1$, $W \ge 2$, and $\alpha_1 < 1/2$, we have $\frac{(1-\alpha_1)W}{A_{i,j}} > 1$. Using the fact that $A_{i,j}$ is at least as large as all entries $A_{i,j'}$ for $j' < j$, we satisfy the conditions to apply the Chernoff bound in Theorem~\ref{thm:chern}. This implies
\[
\Pr[Y_{ij} > W - A_{i,j}] \le \left(\frac{\alpha_1e^{1-\alpha_1}W}{W-A_{i,j}}\right)^{(W-A_{i,j})/A_{i,j}}.
\]
Note that $\frac{W}{W-A_{i,j}} \le 2$ as $W \ge 2$. Because $e^{1-\alpha_1}\le e$ and by the choice of $\alpha_1$, we have
\[
\left(\frac{\alpha_1e^{1-\alpha_1}W}{W-A_{i,j}}\right)^{(W-A_{i,j})/A_{i,j}}
\le
\left(2e\alpha_1\right)^{(W-A_{i,j})/A_{i,j}}
=
\left(\frac{1}{2e^{1/e}\Delta_1}\right)^{(W-A_{i,j})/A_{i,j}}.
\]

Then we prove the final inequality in two parts. First, we see that $W \ge 2$ and $A_{i,j} \le 1$ imply that $\frac{W - A_{i,j}}{A_{i,j}}\ge 1$. This implies
\[
\left(\frac{1}{2\Delta_1}\right)^{(W-1)/A_{i,j}} \le \frac{1}{2\Delta_1}.
\]
Second, we see that
\[
(1/e^{1/e})^{(W-A_{i,j})/A_{i,j}} \le (1/e^{1/e})^{1/A_{i,j}} \le A_{i,j}
\]
for $A_{i,j} \le 1$, where the first inequality holds because $W - A_{i,j} \ge 1$ and the second inequality holds by Lemma~\ref{lem:first-ineq}. This concludes the proof.
\end{proof}

\begin{theorem}\label{thm:W2-weak}
When setting $\alpha _1= \frac{1}{c_1\Delta_1}$ where $c_1 = 4e^{1+1/e}$, $\apxpips(A, b, \alpha_1)$ is a randomized $(\alpha_1/2)$-approximation algorithm for PIPs with width $W \ge 2$.
\end{theorem}
\begin{proof}
Fix $j \in [n]$. By Lemma~\ref{lem:W2-weak} and the definition of $\Delta_1$, we have
\[
 \sum_{i = 1}^m \Pr[E_{ij} | X_j = 1] \le \sum_{i = 1}^m \frac{A_{i,j}}{2\Delta_1} \le \frac{1}{2}.
\]
By Lemma~\ref{lem:general}, which shows that  upper bounding the sum of the rejection probabilities by $\gamma$ for every item leads to an $\alpha_1(1-\gamma)$-approximation, we get the desired result.
\end{proof}

\subsection{An $\Omega(\frac{1}{(1 +\Delta_1/W)^{1/(W-1)}})$-approximation}\label{sec:W2-strong}

We improve the bound from the previous section by setting 
$\alpha_1 = \frac{1} {c_2(1 +\Delta_1/W)^{1/(W-1)}}$ where 
$c_2 = 4e^{1+2/e}$. Note that the scaling factor becomes larger as 
$W$ increases. The analysis of the following lemma is similar to that of 
Lemma~\ref{lem:W2-weak} and is therefore left for the appendix.

\begin{lemma}\label{lem:W2-strong}
Let $\alpha _1= \frac{1}{c_2(1 +\Delta_1/W)^{1/(W-1)}}$ for $c_2 = 4e^{1 + 2/e}$. Let $i \in [m]$ and $j \in [n]$. Then in the algorithm $\apxpips(A, b, \alpha_1)$, we have $\Pr[E_{ij} | X_j = 1] \le \frac{A_{i,j}}{2\Delta_1}$.
\end{lemma}

If we replace Lemma~\ref{lem:W2-weak} with Lemma~\ref{lem:W2-strong} in the proof of Theorem~\ref{thm:W2-weak}, we obtain the following stronger guarantee.

\begin{theorem}\label{thm:W2-strong}
When setting $\alpha _1= \frac{1}{c_2(1 +\Delta_1/W)^{1/(W-1)}}$ where $c_2 = 4e^{1 + 2/e}$, for PIPs with width $W \ge 2$,
$\apxpips(A, b, \alpha_1)$ is a randomized $(\alpha_1 /2)$-approximation.
\end{theorem}

\subsection{A $(1-O(\eps))$-approximation when $W \ge \Omega(\frac{1}{\eps^2}\ln(\frac{\Delta_1}{\eps}))$}\label{sec:largeW}

In this section, we give a randomized $(1-O(\eps))$-approximation for
the case when $W \ge
\Omega(\frac{1}{\eps^2}\ln(\frac{\Delta_1}{\eps}))$. We use the
algorithm $\apxpips$ in Figure~\ref{fig:pseudo-W2} with the scaling
factor $\alpha_1 = 1 - \eps$. The analysis follows the same structure as
the analyses for the lemmas bounding the rejection probabilities from the 
previous sections. The proof can be found in the appendix.

\begin{lemma}\label{lem:bigW}
Let $0 < \eps < \frac{1}{e}$, $\alpha_1 = 1 - \eps$, and $W = \frac{2}{\eps^2}\ln(\frac{\Delta_1}{\eps}) + 1$. Let $i \in[m]$ and $j \in [n]$. Then in $\apxpips(A, b, \alpha_1)$, we have $\Pr[E_{ij} | X_j = 1] \le e\cdot\frac{\eps A_{i,j}}{\Delta_1}$.
\end{lemma}

Lemma~\ref{lem:bigW} implies that we can upper bound the sum of the rejection probabilities for any item $j$ by $e\eps$, leading to the following theorem.

\begin{theorem}
Let $0 < \eps < \frac{1}{e}$ and $W = \frac{2}{\eps^2}\ln(\frac{\Delta_1}{\eps}) + 1$. When setting $\alpha_1 = 1 - \eps$ and $c = e+1$, $\apxpips(A,b,\alpha_1)$ is a randomized $(1 - c\eps)$-approximation algorithm.
\end{theorem}
\begin{proof}
Fix $j \in [n]$. By Lemma~\ref{lem:bigW} and the definition of $\Delta_1$,
\[
\sum_{i =1}^m \Pr[E_{ij} | X_j = 1] \le \sum_{i=1}^m \frac{e\eps A_{i,j}}{\Delta_1} \le e\eps.
\]
By Lemma~\ref{lem:general}, which shows that an upper bound on the rejection probabilities of $\gamma$ leads to an $\alpha_1(1 - \gamma)$-approximation, we have an $\alpha_1(1 - e\eps)$-approximation. Then note that $\alpha_1(1 - e\eps) = (1 - \eps)(1-e\eps) \ge 1 - (e+1)\eps$. This concludes the proof.
\end{proof}

\section{The small width regime: $W = (1+\eps)$}\label{sec:vio}
We now consider the regime when the width is small.  Let $W = 1+ \eps$
for some $\eps \in (0,1]$. We cannot apply the simple sorting based scheme
that we used for the large width regime. We borrow the idea from
\cite{bansal-sparse} in splitting the coordinates into big and small
in each constraint; now the definition is more refined and depends on
$\eps$.  Moreover, the small coordinates and the big coordinates
have their own reserved capacity in the constraint. This is crucial
for the analysis. We provide more formal details below.

We set $\alpha_2$ to be $\frac{\eps^2} {c_3\Delta_1}$ where $c_3 =
8e^{1+2/e}$.
The alteration step differentiates between ``small" and ``big"
coordinates as follows.  For each $i \in [m]$, let $S_i = \{j :
A_{i,j} \le \eps /2\}$ and $B_i = \{j : A_{i,j} > \eps /2\}$. We say
that an index $j$ is \emph{small} for constraint $i$ if $j \in
S_i$. Otherwise we say it is \emph{big} for constraint $i$ when $j \in
B_i$. For each constraint, the algorithm is allowed
to pack a total of $1+\eps$ into that constraint. The algorithm
separately packs small indices and big indices.  In an $\eps$
amount of space, small indices that were chosen in the rounding step
are sorted in increasing order of size and greedily packed until the
constraint is no longer satisfied. The big indices are packed by
arbitrarily choosing one and packing it into the remaining space of
$1$. The rest of the indices are removed to ensure feasibility.
Figure~\ref{fig:pseudo-vio} gives pseudocode for the randomized 
algorithm $\apxvio$ which yields an $\Omega(\eps^2 / \Delta_1)$-approximation.

\begin{figure}[t]
\centering
\fbox{\parbox{0.76\linewidth}
{
$\apxvio(A, b, \eps, \alpha_2)$:
\begin{algorithmic}
\STATE let $x$ be the optimum fractional solution of the natural LP relaxation
\STATE for $j \in [n]$, set $x_j'$ to be $1$ independently with probability $\alpha_2 x_j$ and $0$ otherwise
\STATE $x'' \leftarrow x'$
\FOR{$i \in [m]$}
	\IF{$\abs{S_i} = 0$}
		\STATE $s \leftarrow 0$
	\ELSE
		\STATE sort and renumber such that $A_{i,1} \le \cdots\le A_{i,n}$
		\STATE $s \leftarrow \max\left\{\ell \in S_i : \sum_{j = 1}^{\ell}A_{i,j}x_{j}' \le \eps \right\}$
	\ENDIF
	\STATE if $\abs{B_i} = 0$, then $t = 0$, otherwise let $t$ be an arbitrary element of $B_i$
	\STATE for each $j \in [n]$ such that $j > s$ and $j \ne t$, set $x_j'' = 0$
\ENDFOR
\RETURN{$x''$}
\end{algorithmic}

}}
\caption{By setting the scaling factor $\alpha_2 = \frac{\eps^2}{c\Delta_1}$ for a sufficiently large constant $c$, $\apxvio$ is a randomized $\Omega(\eps^2 / \Delta_1)$-approximation for PIPs with width $W = 1 + \eps$ for some $\eps \in (0,1]$ (see Theorem~\ref{thm:vio}). }
\label{fig:pseudo-vio}
\end{figure}

It remains to bound the rejection probabilities. Recall that for $j \in [n]$, we define $X_j$ to be the indicator random variable $\1(x_j' = 1)$ and $E_{ij}$ is the event that $j$ was rejected by constraint $i$.

We first consider the case when index $j$ is big for constraint $i$. Note that it is possible that there may not exist any big indices for a given constraint. The same holds true for small indices.

\begin{lemma}\label{lem:big-vio}
Let $\eps \in (0,1]$ and $\alpha_2 = \frac{\eps^2}{c_3\Delta_1}$ where $c_3 = 8e^{1+2/e}$. Let $i \in [m]$ and $j \in B_i$. Then in $\apxvio(A, b, \eps, \alpha_2)$, we have $\Pr[E_{ij} | X_j = 1] \le \frac{A_{i,j}}{2\Delta_1}$.
\end{lemma}
\begin{proof}
Let $\mathcal{E}$ be the event that there exists $j' \in B_i$ such that $j' \ne j$ and $X_{j'} = 1$. Observe that if $E_{ij}$ occurs and  $X_j = 1$, then it must be the case that at least one other element of $B_i$ was chosen in the rounding step. Thus,
\[
\Pr[E_{ij} | X_j=1] \le \Pr[\mathcal{E}] \le \sum_{\substack{\ell \in B_i\\\ell \ne j}} \Pr[X_\ell = 1] \le \alpha_2\sum_{\ell \in B_i} x_{\ell},
\]
where the second inequality follows by the union bound. Observe that for all $\ell \in B_i$, we have $A_{i,\ell} > \eps /2$. By the LP constraints, we have $1 + \eps \ge \sum_{\ell \in B_i} A_{i,\ell} x_{\ell} > \frac{\eps}{2}\cdot \sum_{\ell \in B_i} x_{\ell}$. Thus, $\sum_{\ell \in B_i} x_\ell \le \frac{1+\eps}{\eps /2} = 2/\eps + 2$.

Using this upper bound for $\sum_{\ell \in B_i} x_{\ell}$, we have
\[
\alpha_2 \sum_{\ell \in B_i}  x_{\ell}
\le \frac{\eps^2}{c_3\Delta_1}\left(\frac{2}{\eps} + 2\right)
\le \frac{4\eps}{c_3\Delta_1}
\le \frac{A_{i,j}}{2\Delta_1},
\]
where the second inequality utilizes the fact that $\eps \le 1$ and the third inequality holds because $c_3 \ge 16$ and $A_{i,j} > \eps /2$.
\end{proof}

Next we consider the case when index $j$ is small for constraint
$i$. The analysis here is similar to that in the preceding section
with width at least $2$. The proof is left for the appendix.

\begin{lemma}\label{lem:small-vio}
Let $\eps \in (0,1]$ and $\alpha_2 = \frac{\eps^2}{c_3\Delta_1}$ where $c_3 = 8e^{1+2/e}$. Let $i \in [m]$ and $j \in S_i$. Then in  $\apxvio(A, b, \eps, \alpha_2)$, we have $\Pr[E_{ij} | X_j = 1] \le \frac{A_{i,j}}{2\Delta_1}$.
\end{lemma}

As Lemma~\ref{lem:small-vio} shows that the rejection probability is
small, we can prove the following approximation guarantee much like
in Theorems~\ref{thm:W2-weak} and~\ref{thm:W2-strong}.

\begin{theorem}\label{thm:vio}
Let  $\eps \in (0,1]$. When setting $\alpha_2 = \frac{\eps^2}{c_3\Delta_1}$ for $c_3 = 8e^{1+2/e}$, for PIPs with width $W = 1 + \eps$, $\apxvio(A, b, \eps, \alpha_2)$ is a randomized $(\alpha_2/2)$-approximation algorithm .
\end{theorem}
\begin{proof}
Fix $j \in [n]$. Then by Lemmas~\ref{lem:big-vio} and~\ref{lem:small-vio} and the definition of $\Delta_1$, we have
\[
 \sum_{i = 1}^m \Pr[E_{ij} | X_j = 1] \le \sum_{i = 1}^m \frac{A_{i,j}}{2\Delta_1} \le \frac{1}{2}.
\]
Recall that Lemma~\ref{lem:general} gives an $\alpha_2(1-\gamma)$-approximation where $\gamma$ is an upper bound on the sum of the rejection probabilities for any item. This concludes the proof.
\end{proof}

\bibliographystyle{acm}
\bibliography{refs}

\begin{thebibliography}{10}

\bibitem{bansal-sparse}
{\sc Bansal, N., Korula, N., Nagarajan, V., and Srinivasan, A.}
\newblock Solving packing integer programs via randomized rounding with
  alterations.
\newblock {\em Theory of Computing 8}, 24 (2012), 533--565.

\bibitem{chan-knapsack}
{\sc Chan, T.~M.}
\newblock Approximation schemes for 0-1 knapsack.
\newblock In {\em 1st Symposium on Simplicity in Algorithms\/} (2018).

\bibitem{CK04}
{\sc Chekuri, C., and Khanna, S.}
\newblock On multidimensional packing problems.
\newblock {\em SIAM journal on computing 33}, 4 (2004), 837--851.

\bibitem{cq}
{\sc Chekuri, C., and Quanrud, K.}
\newblock On approximating (sparse) covering integer programs.
\newblock In {\em Proceedings of the Thirtieth Annual ACM-SIAM Symposium on
  Discrete Algorithms\/} (2019), SIAM, pp.~1596--1615.

\bibitem{cvz-crs}
{\sc Chekuri, C., Vondr{\'a}k, J., and Zenklusen, R.}
\newblock Submodular function maximization via the multilinear relaxation and
  contention resolution schemes.
\newblock {\em SIAM Journal on Computing 43}, 6 (2014), 1831--1879.

\bibitem{chs-16}
{\sc Chen, A., Harris, D.~G., and Srinivasan, A.}
\newblock Partial resampling to approximate covering integer programs.
\newblock In {\em Proceedings of the twenty-seventh annual ACM-SIAM symposium
  on Discrete algorithms\/} (2016), Society for Industrial and Applied
  Mathematics, pp.~1984--2003.

\bibitem{FriezeC84}
{\sc Frieze, A., and Clarke, M.}
\newblock Approximation algorithms for the m-dimensional 0-1 knapsack problem:
  worst-case and probabilistic analyses.
\newblock {\em European Journal of Operational Research 15}, 1 (1984),
  100--109.

\bibitem{h-15}
{\sc Harvey, N.~J.}
\newblock A note on the discrepancy of matrices with bounded row and column
  sums.
\newblock {\em Discrete Mathematics 338}, 4 (2015), 517--521.

\bibitem{hastad-MIS}
{\sc H{\aa}stad, J.}
\newblock Clique is hard to approximate within $n^{1- \epsilon}$.
\newblock {\em Acta Mathematica 182}, 1 (1999), 105--142.

\bibitem{hss-06}
{\sc Hazan, E., Safra, S., and Schwartz, O.}
\newblock On the complexity of approximating k-set packing.
\newblock {\em Computational Complexity 15}, 1 (2006), 20--39.

\bibitem{mitz-upfal-rand}
{\sc Mitzenmacher, M., and Upfal, E.}
\newblock {\em Probability and computing: Randomized algorithms and
  probabilistic analysis}.
\newblock Cambridge university press, 2005.

\bibitem{Pritchard09}
{\sc Pritchard, D.}
\newblock Approximability of sparse integer programs.
\newblock In {\em European Symposium on Algorithms\/} (2009), Springer,
  pp.~83--94.

\bibitem{PritchardC11}
{\sc Pritchard, D., and Chakrabarty, D.}
\newblock Approximability of sparse integer programs.
\newblock {\em Algorithmica 61}, 1 (Sep 2011), 75--93.

\bibitem{RaghavanT87}
{\sc Raghavan, P., and Thompson, C.~D.}
\newblock Randomized rounding: a technique for provably good algorithms and
  algorithmic proofs.
\newblock {\em Combinatorica 7}, 4 (1987), 365--374.

\bibitem{Srinivasan99}
{\sc Srinivasan, A.}
\newblock Improved approximation guarantees for packing and covering integer
  programs.
\newblock {\em SIAM Journal on Computing 29}, 2 (1999), 648--670.

\bibitem{t-01}
{\sc Trevisan, L.}
\newblock Non-approximability results for optimization problems on bounded
  degree instances.
\newblock In {\em Proceedings of the thirty-third annual ACM symposium on
  Theory of computing\/} (2001), ACM, pp.~453--461.

\bibitem{zuckerman-MIS}
{\sc Zuckerman, D.}
\newblock Linear degree extractors and the inapproximability of max clique and
  chromatic number.
\newblock In {\em Proceedings of the thirty-eighth annual ACM symposium on
  Theory of computing\/} (2006), ACM, pp.~681--690.

\end{thebibliography}

\section*{Appendix}
\appendix
\section{Chernoff Bounds and Useful Inequalities}
  The following standard Chernoff bound is used to obtain a more 
  convenient Chernoff bound in Theorem~\ref{thm:chern}. The proof of 
  Theorem~\ref{thm:chern} follows directly from choosing $\delta$ such 
  that $(1 + \delta)\mu = W - \beta$ and applying 
  Theorem~\ref{thm:std-chern}. We include the proof for convenience.

  \begin{theorem}[\cite{mitz-upfal-rand}]\label{thm:std-chern}
    Let $X_1,\ldots, X_n$ be independent random variables where 
    $X_i$ is defined on $\{0, \beta_i\}$, where $0 < \beta_i \le \beta \le 1$ 
    for some $\beta$. Let $X = \sum_i X_i$ and denote $\E[X]$ as $\mu$. 
    Then for any $\delta > 0$,
    \[
      \Pr[X \ge (1 + \delta)\mu] 
      \le 
      \left(\frac{e^\delta}{(1+\delta)^{1+\delta}}\right)^{\mu/\beta}
    \]
  \end{theorem}

  \begin{theorem}\label{thm:chern}
    Let $X_1,\ldots, X_n \in [0,\beta]$ be independent random variables 
    for some $0 < \beta \le 1$. Suppose 
    $\mu = \E[\sum_i X_i] \le \alpha W$ for some $0 < \alpha < 1$ and 
    $W \ge 1$ where $(1 - \alpha)W > \beta$. Then
    \[
      \Pr\left[\sum_i X_i > W - \beta\right] 
      \le 
      \left(\frac{\alpha e^{1-\alpha} W}{W - \beta}\right)^{(W-\beta)/\beta}.
    \]
  \end{theorem}
  \begin{proof}
    Since the right-hand side is increasing in $\alpha$, it suffices to 
    assume $\mu = \alpha W$. Choose $\delta$ such that 
    $(1+\delta)\mu = W - \beta$. Then $\delta = (W - \beta - \mu)/\mu$. 
    Because $\mu = \alpha W$ and since $(1 - \alpha)W > \beta$, we 
    have $\delta = ((1-\alpha)W - \beta)/\mu > 0$. We apply the standard 
    Chernoff bound in Theorem~\ref{thm:chern} to obtain
    \[
      \Pr\left[\sum_iX_i > W- \beta\right] 
      = 
      \Pr\left[\sum_i X_i > (1+\delta)\mu\right] 
      \le 
      \left(\frac{e^\delta}{(1+\delta)^{1+\delta}}\right)^{\mu/\beta}.
    \]
    Because $1 + \delta = (W - \beta)/\mu$ and 
    $\delta = (W - \beta - \mu)/\mu$,
    \[
      \left(\frac{e^\delta}{(1+\delta)^{1+\delta}}\right)^{\mu/\beta}
      = 
      \left(\frac{e^{W - \beta - \mu}}
        {((W - \beta)/\mu)^{W-\beta}}\right)^{1/\beta}.
    \]
    Exponentiating the denominator,
    \[
      \left(\frac{e^{W - \beta - \mu}}
        {((W - \beta)/\mu)^{W-\beta}}\right)^{1/\beta}
      = \exp\left(\frac{1}{\beta} \left(W - \beta - \mu + 
        (W-\beta)\ln\left( \frac{\mu}{W-\beta} \right) \right) \right)
    \]
    As $\mu = \alpha W$,
    \begin{multline*}
      \exp\left(\frac{1}{\beta} \left(W - \beta - \mu + 
        (W-\beta)\ln\left( \frac{\mu}{W-\beta} \right) \right) \right) 
      = \exp\left(\frac{1}{\beta} \left((1-\alpha)W - \beta  
        + (W-\beta)\ln\left( \frac{\alpha W}{W-\beta} \right) \right) \right)
    \end{multline*}
    We can rewrite the exponent to show that
    \[
      \exp\left(\frac{1}{\beta} \left((1-\alpha)W - \beta 
        - (W-\beta)\ln\left( \frac{W-\beta}{\alpha W} \right) \right) \right)
      \le \left(\frac{\alpha e^{1-\alpha}W}
        {W - \beta}\right)^{(W-\beta)/\beta}.
    \]
  \end{proof}

  The following three lemmas are used in the proofs bounding the
  rejection probabilities for different regimes of width. The
  inequalities are easily verified via calculus. The proofs are included
  for the sake of completeness. 

  \begin{lemma}\label{lem:first-ineq}
    Let $x \in (0,1]$. Then $(1/e^{1/e})^{1/x} \le x$.
  \end{lemma}
  \begin{proof}
    Taking logs of both sides of the stated inequality and rearranging,
    it suffices to show that $\ln (1/e^{1/e}) \le x\ln x$ for 
    $x > 0$. $x \ln x$ is convex and its minimum is $-1/e$ at 
    $x = 1/e$. Since $\ln (1/e^{1/e}) = -1/e$, the inequality holds.
  \end{proof}
  
  \begin{lemma}\label{lem:ineq-3}
    Let $y \ge 2$ and $x \in (0, 1]$ . Then $x/y \ge (1/e^{2/e})^{y/2x}$.
  \end{lemma}
  \begin{proof}
    We start with a simple rewriting of the statement. After taking logs 
    and rearranging, it is sufficient to show
    \[
      (x/y) \ln (x/y) 
      \ge 
      (1 / 2)\ln (1/e^{2/e})
      = -1/e.
    \]
    Replacing $x/y$ with $z$, we see that it suffices the prove
    $z \ln z \ge -1/e$ for $0 < z \le 1/2$. We note that $x \ln x$ is convex 
    and its minimum is $-1/e$ at $x = 1/e$. Thus, $z \ln z \ge -1/e$. This 
    concludes the proof.
  \end{proof}
  
  \begin{lemma}\label{lem:exp-ineq}
    Let $0 < \eps \le 1$ and  $x \in (0, 1]$. Then 
    $\eps x /2 \ge (\eps/e^{2/e})^{1/x}$.
  \end{lemma}
  \begin{proof}
    To start, let $d = e^{2/e}/2$ and observe that $d > 1$. We first do a 
    change of variables, replacing $\eps / 2$ with $\eps$ and $x$ with 
    $x/\eps$. If we take a $\log$ of both sides, then our reformulated 
    goal is to show that
    \[
      x\ln x 
      \ge 
      \eps\ln(\eps/d)
    \]
    for $0 < \eps \le 1/2$ and $x \in (0,\eps]$. Letting $f(y) = y \ln y$ 
    and $g(y) = y \ln (y/d)$, we want to show that $f(x) \ge g(\eps)$. 
    We will proceed by cases.

    First, suppose $0 < \eps \le d/e$. It is easy to show that $f$ is 
    decreasing on $(0, 1/e]$ and increasing on $[1/e, \infty)$ and that 
    $g$ is decreasing on $(0,d/e]$ and increasing on $[d/e, \infty)$. As 
    $f$ is decreasing on $(0,1/e]$, for $0 < \eps \le 1/e$, we have 
    $f(x) \ge f(\eps)$ as $x \le \eps$. As $d > 1$, it follows that 
    $f(\eps) \ge g(\eps)$. Therefore, $f(x) \ge g(\eps)$ for 
    $0 < \eps \le 1/e$. Furthermore, as $g$ is decreasing on $[1/e,d/e]$
    and $f$ is increasing on $[1/e,d/e]$, we have $f(x) \ge g(\eps)$ for 
    $0 < \eps \le d/e$.

    For the second case, suppose $d/e < \eps \le 1/2$. Note that the 
    minimum of $f$ on the interval $(0, 1/2]$ is $f(1/e)= -1/e$. Thus, it 
    would suffice to show that $g(\eps) \le -1/e$. As we noted previously 
    that $g$ is increasing on $[d/e, 1/2]$, it would suffice to show that 
    $g(1/2) \le -1/e$. By definition of $g$, we see $g(1/2) = -1/e$. 
    Therefore, $f(x) \ge g(\eps)$. This concludes the proof.
  \end{proof}

\section{Skipped Proofs}
  \subsection{Proof of Lemma~\ref{lem:W2-strong}}
    \begin{proof}
      The proof proceeds similarly to the proof of 
      Lemma~\ref{lem:W2-weak}. Since $\alpha_1 < 1/2$, everything up 
      to and including the application of the Chernoff bound there 
      applies. This gives that for each $i\in[m]$ and $j \in [n]$,
      \[
        \Pr[E_{ij} | X_j = 1] \le \left(2e\alpha_1\right)^{(W-A_{i,j})/A_{i,j}}.
      \]
      By choice of $\alpha_1$, we have
      \[
        \left(2e\alpha_1\right)^{(W-A_{i,j})/A_{i,j}} 
        =  
        \left(\frac{1}{2e^{2/e}(1 +\Delta_1/W)^{1/(W-1)}}\right)
          ^{(W-A_{i,j})/A_{i,j}}
      \]
      We prove the final inequality in two parts. First, note that 
      $\frac{W - A_{i,j}}{A_{i,j}} \ge W - 1$ since $A_{i,j} \le 1$. Thus,
      \[
        \left(\frac{1}{2(1 +\Delta_1/W)^{1/(W-1)}}\right)
          ^{(W-A_{i,j})/A_{i,j}}
        \le \frac{1}{2^{W-1}(1 + \Delta_1/W)}
        \le \frac{W}{2\Delta_1}.
      \]
      Second, we see that
      \[
        \left(\frac{1}{e^{2/e}}\right)^{(W-A_{i,j})/A_{i,j}}
        \le 
        \left(\frac{1}{e^{2/e}}\right)^{W/2A_{i,j}}
        \le  
        \frac{A_{i,j}}{W}
      \]
      for $A_{i,j} \le 1$, where the first inequality holds because $W \ge 2$ 
      and the second inequality holds by Lemma~\ref{lem:ineq-3}.
    \end{proof}

  \subsection{Proof of Lemma~\ref{lem:bigW}}
    \begin{proof}
      Renumber indices so that $A_{i,1} \le \cdots \le A_{i,n}$ and if the 
      index of $j$ changes to $j'$, we still refer to $j'$ as $j$. Let 
      $Y_{ij} = \sum_{\ell = 1}^{j-1} A_{i,\ell} \xi_\ell$ where $\xi_\ell = 1$ if 
      $x_\ell' = 1$ and $0$ otherwise. We first note that
      \[
        \Pr[E_{ij} | X_j = 1] \le \Pr[Y_{ij} > W - A_{i,j}].
      \]
      
      By the choice of $\alpha_1$ and the fact that $A_{i,j} \le 1$ and 
      $W = \frac{2}{\eps^2}\ln(\frac{\Delta_1}{\eps}) + 1$, we have 
      $((1- \alpha_1)W)/A_{i,j} \ge \eps W = \frac{2}{\eps}
      \ln(\frac{\Delta_1}{\eps}) + \eps$. A direct argument via calculus 
      shows $\frac{2}{\eps}\ln(\frac{\Delta_1}{\eps}) + \eps > 1$ for 
      $\eps \in (0,\frac{1}{e})$. Thus, $(1-\alpha_1)W > A_{i,j}$.
      
      By the LP constraints, $\E[Y_{ij}] \le \alpha_1 W$. Then as 
      $A_{i,j'} \le A_{i,j}$ for all $j' < j$, we can apply the Chernoff bound 
      in Theorem~\ref{thm:chern} to obtain
      \[
        \Pr[Y_{ij} \ge W - A_{i,j}] 
        \le 
        \left(\frac{\alpha_1 e^{1-\alpha_1} W}{W - A_{i,j}}\right)
          ^{(W-A_{i,j})/A_{i,j}}.
      \]
      As $A_{i,j} \le 1$,
      \[
        \left(\frac{W}{W- A_{i,j}}\right)^{(W-A_{i,j})/A_{i,j}} 
        \le 
        \left(\frac{W}{W- 1}\right)^{W - 1} \le e,
      \]
      where the last inequality follows from the fact that 
      $(1-1/z)^{z-1} \ge 1/e$ for all $z \ge 1$. Then
      \[
        \left(\frac{\alpha_1 e^{1-\alpha_1} W}{W - A_{i,j}}\right)
          ^{(W-A_{i,j})/A_{i,j}}
        \le 
        e\cdot\left(\alpha_1 e^{1-\alpha_1}\right)^{(W-A_{i,j})/A_{i,j}}.
      \]
      By the choice of $\alpha_1$,
      \[
        e\cdot\left(\alpha_1 e^{1-\alpha_1}\right)^{(W-A_{i,j})/A_{i,j}}
        =
        e\cdot\left((1-\eps)e^{\eps}\right)^{(W-A_{i,j})/A_{i,j}}.
      \]
      For $0 < \eps < \frac{1}{e}$, we have 
      $1 - \eps \le \exp(-\eps - \frac{\eps^2}{2})$. As 
      $W = \frac{2}{\eps^2}\ln(\frac{\Delta_1}{\eps}) + 1$ and 
      $A_{i,j} \le 1$,
      \[
        e\cdot\left((1-\eps) e^\eps\right)^{(W-A_{i,j})/A_{i,j}}
        \le 
        e\cdot\left(e^{-\eps^2/2}\right)^{\frac{2}{\eps^2} 
          \ln(\frac{\Delta_1}{\eps})}
        \le 
        e\cdot\exp\left(-\frac{\ln(\frac{\Delta_1}{\eps})}{A_{i,j}}\right).
      \]
      Observe that $\frac{1}{A_{i,j}} - \ln(\frac{e}{A_{i,j}}) \ge 0$. For 
      $A_{i,j} \in [0,1]$, a direct argument shows 
      $\frac{\ln(t)}{A_{i,j}} - \ln(\frac{t}{A_{i,j}})$ is increasing in $t$ for 
      $t \ge e$. As $\Delta_1 / \eps > e$, we have 
      $\frac{\ln(\frac{\Delta_1}{\eps})}{A_{i,j}} \ge \ln(\frac{\Delta_1}
      {\eps A_{i,j}})$. Therefore,
      \[
        e\exp\left(-\frac{\ln(\frac{\Delta_1}{\eps})}{A_{i,j}}\right)
        \le 
        e\exp\left(-\ln\left(\frac{\Delta_1}{\eps A_{i,j}}\right)\right)
        = 
        \frac{e\eps A_{i,j}}{\Delta_1}.
      \]
      This concludes the proof.
\end{proof}

\subsection{Proof of Lemma~\ref{lem:small-vio}}
  \begin{proof}
    Renumber the indices so that $A_{i,1} \le \cdots \le A_{i,n}$. Note 
    that the index $j$ might have changed to $j'$ but for simplicity we
    refer to $j'$ as $j$. Let $\xi_\ell = 1$ if $x_\ell' = 1$ and $0$ 
    otherwise. Let $Y_{ij} = \sum_{\ell = 1}^{j-1} A_{i,\ell} \xi_\ell$. 
    We have
    \[
      \Pr[E_{ij} | X_j = 1] \le \Pr[Y_{ij} \ge \eps - A_{i,j}].
    \]

    Let $A_{i,\ell}' = \frac{2}{\eps}\cdot A_{i,\ell}$ for $\ell \in [j]$. As 
    $A_{i,\ell} \le \eps/2$ for all $\ell \in [j]$, we have $A_{i,\ell}'\in [0,1]$. 
    Let $Y_{ij}' = \sum_{\ell = 1}^{j-1}A_{i,\ell}' \xi_\ell$. Then
    \[
      \Pr[Y_{ij} \ge \eps - A_{i,j}] = \Pr[Y_{ij}' \ge 2 - A_{i,j}'].
    \]
    To upper bound $\E[Y_{ij}']$, we have
    \[
      \E[Y_{ij}'] 
      = 
      \sum_{\ell = 1}^{j-1}A_{i,\ell}'\cdot \Pr[X_\ell = 1] 
      \le 
      \frac{2\alpha_2}{\eps}\sum_{\ell = 1}^{n} A_{i,\ell}x_\ell 
      \le 
      \frac{2\alpha_2(1+\eps)}{\eps} 
      = 
      \frac{2\eps(1+\eps)}{c_3\Delta_1}.
    \]
    Let $\alpha_2' = \frac{2\eps}{c_3\Delta_1}$ and $W = 2$. Then 
    $\E[Y_{ij}'] \le \alpha_2'W$. As $\alpha_2' < 1/2$ and $A_{i,j}' \le 1$, 
    we see that $((1-\alpha)W)/A_{i,j}' > 1$. Therefore, as 
    $A_{i,\ell}' \le A_{i,j}'$ for all $\ell < j$, we can apply the Chernoff 
    bound in Theorem~\ref{thm:chern} to obtain
    \[
      \Pr[Y_{ij}' \ge 2 - A_{i,j}'] 
      \le 
      \left(\frac{\alpha_2'e^{1-\alpha_2'}W}
        {W - A_{i,j}'}\right)^{(W-A_{i,j}')/A_{i,j}'}.
    \]
    Observe that $e^{1-\alpha_2'} \le e$ and $\frac{W}{W- A_{i,j}'} \le 2$ 
    since $W = 2$ and $A_{i,j}' \le 1$. By our choice of $\alpha_2'$,
    \[
      \left(\frac{\alpha_2'e^{1-\alpha_2'}W}
        {W - A_{i,j}'}\right)^{(W-A_{i,j}')/A_{i,j}'}
      \le 
      \left(2e\alpha_2'\right)^{(W-A_{i,j}')/A_{i,j}'}
      =
      \left(\frac{\eps}{2e^{2/e}\Delta_1}\right)^{(W-A_{i,j}')/A_{i,j}'}
    \]

    We prove the final inequality in two parts. First, we note that 
    $\frac{W-A_{i,j}'}{A_{i,j}'} \ge 1$ since $W = 2$ and $A_{i,j}' \le 1$. 
    Then
    \[
      \left(\frac{1}{2\Delta_1}\right)^{(W-A_{i,j}')/A_{i,j}'} 
      \le 
      \frac{1}{2\Delta_1}.
    \]
    Second, we observe $\frac{W-A_{i,j}'}{A_{i,j}'} \ge 1/A_{i,j}'$ since 
    $W = 2$ and $A_{i,j}' \le 1$. Then we can apply 
    Lemma~\ref{lem:exp-ineq} to obtain
    \[
      (\eps/e^{2/e})^{(W-A_{i,j}')/A_{i,j}'} 
      \le 
      (\eps/e^{2/e})^{1/A_{i,j}'}
      \le 
      \frac{\eps A_{i,j}'}{2}.
    \]
    We have shown $\Pr[E_{ij} | X_j  = 1]\le \frac{\eps A_{i,j}'}{4\Delta_1}$. 
    Since $A_{i,j}' = A_{i,j}\cdot \frac{2}{\eps}$, the result follows.
  \end{proof}
\end{document}